\documentclass{article}

\PassOptionsToPackage{numbers, compress}{natbib}
\usepackage[preprint]{NFFL2021}
\usepackage{graphicx}
\graphicspath{ {./images/} }
\usepackage{subcaption}

\usepackage[small,compact]{titlesec} 
\titlespacing*{\section}{0pt}{0.1\baselineskip}{.01\baselineskip}
\titlespacing*{\subsection}{0pt}{.1\baselineskip}{.01\baselineskip}
\titlespacing{\paragraph}{0pt}{.10\baselineskip}{.05\baselineskip}%{50pt}
\usepackage{comment}
\usepackage[utf8]{inputenc} % allow utf-8 input
\usepackage[T1]{fontenc}    % use 8-bit T1 fonts
\usepackage{hyperref}       % hyperlinks
\usepackage{url}            % simple URL typesetting
\usepackage{booktabs}       % professional-quality tables
\usepackage{amsfonts}       % blackboard math symbols
\usepackage{nicefrac}       % compact symbols for 1/2, etc.
\usepackage{microtype}      % microtypography
\usepackage{xcolor}         % colors

\usepackage{algorithm}
\usepackage{amsmath}
\usepackage{amsthm}
\usepackage{algorithmic}
\usepackage{array}
\newtheorem{theorem}{Theorem}
\newtheorem{lemma}[theorem]{Lemma}

\theoremstyle{definition}
\newtheorem{definition}{Definition}[section]

\newtheorem{assumption}{Assumption}[section]
\newtheorem*{remark}{Remark}
\newcommand{\M}{\mathcal{M}}

\title{Secure Byzantine-Robust Distributed Learning via Clustering}
\author{%
  Raj Kiriti Velicheti\\
  Coordinated Sciences Laboratory\\
  University of Illinois at Urbana-Champaign\\
  \texttt{rkv4@illinois.edu} \\
  \And
  Derek Xia\\
  Department of Computer Science\\
  University of Illinois at Urbana-Champaign\\
  \texttt{derekx3@illinois.edu} \\
  \And
  Oluwasanmi Koyejo\\
  Department of Computer Science\\
  University of Illinois at Urbana-Champaign\\
  \texttt{sanmi@illinois.edu} \\
}
\begin{document}

\maketitle

\begin{abstract}%\rkv{shortened this a bit} \sk{shortened some more}
    %With the abundance of data collected remotely through edge devices, paradigms such as federated learning have grown in popularity. 
    %Although federated learning can enhance user privacy by restricting the data to remain on the edge devices, malicious updates can cause convergence issues. 
    Federated learning systems that jointly preserve Byzantine robustness and privacy have remained an open problem. Robust aggregation, the standard defense for Byzantine attacks, generally requires server access to individual updates or nonlinear computation -- thus is incompatible with privacy-preserving methods such as secure aggregation via multiparty computation. To this end, we propose SHARE (Secure Hierarchical Robust Aggregation), a distributed learning framework designed to cryptographically preserve client update privacy and robustness to Byzantine adversaries simultaneously. The key idea is to incorporate secure averaging among randomly clustered clients before filtering malicious updates through robust aggregation. Experiments show that SHARE has similar robustness guarantees as existing techniques while enhancing privacy.
\end{abstract}

\section{Introduction}\label{Sec: Intro}
An increasing amount of data is being collected in a decentralized manner on devices across institutions\cite{kairouz2019advances}. %through internet-of-things (IoT) and edge devices We also highlight that federated learning methods (including ours) are also deployed for multi-silo federated learning between a few (typically 10’s) in institutions that participate in every round\sk{If we wan to argue multi-silo applications later, I suggest multi-silo examples (as we did in the rebuttal)}. 
Traditionally, machine learning with such devices require centralized data collection, which increases communication costs while posing a threat to privacy, especially when these devices gather personal user data. Distributed learning frameworks like federated learning attempt to address these issues by sharing model updates from client devices, rather than data, to a centralized server~\citep{konevcny2016federated,kairouz2019advances,hard2018federated,li2020federated}. 

% An increasing amount of data is continuously collected in a decentralized manner through internet-of-things (IoT) and edge devices. Traditionally,  machine learning insights with such devices require a centralized collection of data. Such centralized data collection increases communication costs while posing a threat to privacy, especially when these devices gather personal user data. To this end, distributed learning frameworks like federated learning have emerged as a popular approach to overcome these barriers. The idea of such architectures is to learn a shared model through repeated synchronization of model updates from individual devices which train the model locally~\citep{konevcny2016federated,kairouz2019advances,hard2018federated,li2020federated}. 

Among the most popular implementations of federated learning is Federated Averaging~\citep{mcmahan2017communication}. While the central coordinating server follows a designated aggregation protocol, the required communication can pose a privacy threat when the system is compromised by a malicious external agent leaking individual model updates. To this end, \citet{bonawitz2017practical} proposed a secure averaging oracle that masks individual client updates such that the server learns their average alone. Nevertheless, since the collaboratively learned model update includes the contribution of all participating clients, benign averaging might fall prey to incorrect device updates either due to arbitrary failures or maliciously crafted updates preventing the devices from learning a good model. 

In recent years, federated learning robustness to Byzantine failures  (i.e., worst-case adversarial coordinated training-time attacks) has gained attention. However, existing robust aggregation techniques require sophisticated nonlinear operations~\citep{xie2019slsgd,xie2020zeno++,blanchard2017machine}, sometimes with server access to individual model updates in the clear -- thus leading to privacy loss. These nonlinear operations adversely affect privacy since privacy-preserving methods such as secure Multi-Party Computation~(MPC) are inefficient for nonlinear operations~\citep{bonawitz2017practical}. This observation highlights a fundamental tension between existing solutions to the two critical problems of privacy and robustness.

We prose a novel hierarchical framework that decouples MPC-based privacy and Byzantine robustness protection mechanisms in this work. The basic idea is to implement a secure averaging oracle among randomly clustered clients, then filtering these updates using robust aggregation. This approach reveals only the cluster averaged update to the server, thus can help preserve privacy. Simultaneously, the second level of robust aggregation helps to maintain Byzantine robustness.

Taken together, this manuscript proposes a federated learning architecture that preserves security and privacy jointly, thus addressing this gap in the literature. Due to the hierarchical approach, existing robust distributed learning frameworks~\citep{xie2019slsgd,xie2020zeno++,blanchard2017machine} and non-robust secure distributed learning frameworks like \citep{bonawitz2017practical,bonawitz2019towards} can be considered special cases of our proposed approach. To the best of our knowledge, ours is the first approach that scalably combines Byzantine-robustness with privacy using the common single-server architecture. {\bf Summary of contributions:} We propose SHARE; a robust distributed learning framework which flexibly incorporates any Byzantine-robust defenses while enhancing privacy in a single server systems setting. We extend existing theoretical guarantees of robust aggregation oracles to the SHARE framework. Further, we present empirical evaluation of SHARE on benchmark datasets. 
\begin{figure}[h]
\centering
\includegraphics[scale=0.35]{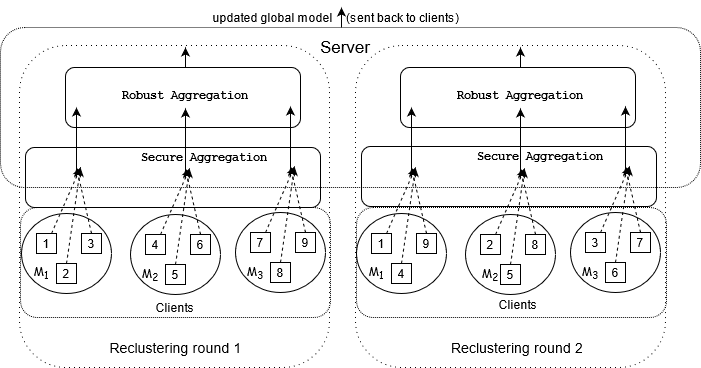}
\caption{This figure illustrates our proposed framework SHARE. Each global round consists of multiple reclustering rounds, updates from which are averaged to obtain the final model update. In each reclustering round (shown by dotted rectangle), updates from clients (numbered squares) are clustered randomly ($\mathcal{M}_1,\mathcal{M}_2,\mathcal{M}_3$), then averaged followed by robust aggregation.} %These aggregate updates (dotted lines) are combined by the server.
%, which are averaged within each cluster to reveal the client clustered updates for robust aggregation at the server.}
\label{fig:shareillus}
\end{figure} 
\section{Related Work}
\label{Sec: Relatedwork}
Byzantine robustness and secure aggregation in distributed learning both have a large existing literature -- though often from different (somewhat disconnected) communities. %Hence we categorize the following survey accordingly. 
\paragraph{Robust Aggregation.} Robustness to Byzantine adversaries is a well-studied problem in distributed and federated learning~\citep{lamport2019Byzantine}. Broadly, existing defenses can be categorized into distance-based robust aggregation or validation data-based aggregation. The general idea behind distance-based metrics is to find an update closer to the benign mean in $l_2$ norm distance. \citet{xie2019slsgd} suggest utilizing coordinate-wise trimmed mean. \citet{blanchard2017machine} choose a model update closest to most other updates.% \citet{chen2018draco} uses redundancy of gradients followed by majority vote. 
%\citet{mhamdi2018hidden} enhances robustness by combining defenses. 
\citet{ghosh2019robust} suggest optimal statistical rates utilizing median and trimmed mean. All these distance-based defenses require a majority of the clients participating in the protocol to be benign. On the other hand, validation data-based aggregation defenses such as Zeno~\citep{xie2019zeno,xie2020zeno++} perform suspicion-based aggregation based on a score evaluated on validation data held at the server. These methods can tolerate arbitrary Byzantine poisoning. All the above techniques require the server to see the local model updates in the clear, posing a privacy threat.
%Most of the above methods are not well validated for heterogeneous data,\sk{I am still unconvinced that we understand how this approach improved heterogeneous data. Careful about over-claiming here.} which is of interest in federated learning settings. \citet{he2020Byzantine} suggest re-sampling before robust aggregation to homogenize gradients. 

\paragraph{Privacy via Secure Aggregation.} In many distributed learning settings, the secure computation boils down to computing a secure average. \citet{bonawitz2017practical} utilizes a pairwise secret share to achieve the same. %To enhance scalability and minimizing communication, \citet{bonawitz2019towards} suggest a two-step aggregation procedure that is easily integrated with our proposed approach. %\citet{10.1145/3372297.3417885} suggest a modification to this scheme, further reducing communication overhead for each client from linear in the number of clients to polylogarithmic. 
\citet{aono2017privacy} follow a slightly different approach and utilize additively homomorphic encryption for secure update computation. While these methods work well with linear aggregation methods, extending secure multiparty computation to non-linear robust aggregation schemes introduces additional computational and communication overhead, quickly becoming impractical for real computational loads.
%\begin{center}

\citet{he2020secure, pillutla2019robust,wang2021bidirectional} are the closest to our work in the sense that they are proposed to address the problem of robustness and privacy in distributed learning jointly. Compared to \citet{he2020secure}, which requires two non-colluding servers, we achieve this with a single server, which may be a more realistic architecture for practical use cases. Further, \citet{he2020secure} tailor their approach to distance-based robust aggregation. In contrast, our proposed approach is easily combined with most existing Byzantine-robust aggregation schemes, including filtering-based defenses such as Zeno++~\citep{xie2020zeno++}. On the other hand, \citet{pillutla2019robust} reduce the filtering computation of the median into a sequence of linear computations. Unfortunately, this approach requires that the Byzantine device follows the computational protocol over multiple rounds, which is a strong assumption in practice. Thus, while inspired by Byzantine tolerance, \citet{pillutla2019robust} do not claim Byzantine robustness.
%\rkv{Finally, our work subsumes \citet{wang2021bidirectional} and hence provides flexibility to plug-in any existing robustness techniques. We emphasize that unlike \citet{wang2021bidirectional}, our intention was to develop a wrapper approach that can be combined with any robustness method, with analysis and performance depending on the underlying robust aggregator. Further, we address the problem of signal loss due to clustering with an important reclustering step. Finally, we provide a clear intuition wherever possible on the trade-off between privacy and breaking point of robust aggregators}
After completing this work, we were made aware of a related approach \citep{wang2021bidirectional} using a hierarchical architecture with a robust mean aggregator. Compared to \citet{wang2021bidirectional}, our approach is a wrapper method that can be combined with any robust aggregator, with analysis and performance depending on the choice of aggregator (e.g., we analyze and compare trimmed mean, krum, Zeno++). Further, we address the problem of signal loss due to clustering via a novel reclustering step. %Finally, we provide a clear intuition wherever possible on the trade-off between privacy and breaking point of robust aggregators

% \sk{Please see how I use citet (for citations that use author name inline with text) and citep (for citations that are ref only), and fix the rest of the doc.}

%\end{center}
\section{Problem Formulation}
\label{Sec: Problem}
We consider the optimization problem $\min_{x\in\mathbb{R}^d} F(x)$ where, $F(x)=\frac{1}{n}\sum_{i\in[n]}\mathbb{E}_{z_i\sim \mathcal{D}_i}f_i(x;z_i)$, hence the goal is to learn a model $x$ which performs well on average using $z_i$ sampled from local data distribution $\mathcal{D}_i,\forall i\in[n]$. The notations used in this paper are summarized in Table \ref{tab: notations} (Appendix A).

This problem is solved in a distributed and iterative manner. In each global iteration ($t<T$), sampled clients compute a private model update ($\Delta x_i^K$) by running multiple steps (K-steps) of Stochastic Gradient Descent~(SGD) on the local data available (${z_i\sim\mathcal{D}_i}$). Then server can computes a global model update. For instance, when using simple averaging, the server update is $x^{t}=x^{t-1}+\eta\sum_{i\in[n]}\Delta x_i^K$, where $\eta$ is global learning rate. %The proposed approach involves combining subsets of clients into clusters. 
We consider the following privacy and security threats:%(\rkv{revised these a bit, please take a look})
\begin{itemize}
\item {\em Privacy threat model}: We consider an honest but curious server. This specification allows the server to interpret the device data from the updates, hence breaching privacy. The assumption of honest server implies that the server still follows the underlying protocol. % \sk{is there a privacy threat from the client side? Maybe no...} 
\item {\em Robustness threat model}: We consider a fixed (unknown) subset ($q$) of machines that can co-ordinate and send arbitrary updates to the server hence deviating from the intended distributed learning protocol.
\end{itemize}
%It can be seen that in worst case failure mode, the Byzantine clients can send crafted updates to break convergence. Additionally, Byzantine clients can become curious and intercept messages from other clients in a cluster, hence when enough clients are Byzantine, they might be able to decipher every client update.
\section{Methodology}
\label{Sec: Solution}
We propose two-step hierarchical aggregation SHARE (Secure Hierarchical Robust Aggregation) as a defense against the specified robustness and privacy threat models. In particular, our approach allows a decoupling of the security and robustness into two steps (as illustrated in Figure \ref{fig:shareillus}). First, in every global epoch, all participating clients are clustered randomly into groups. Clients within each group share pairwise secret keys and utilize them to mask their individual updates such that the server only learns the average within the cluster. This ensures client update privacy. These client cluster updates are then filtered using Byzantine-robust aggregation techniques. Further, we can repeat this process multiple times in a global epoch to aid in reducing variance. The detailed algorithm is outlined in Algorithm \ref{algo:Share}. Without loss of generality, we assume clusters of uniform size. %\sk{it's not clear to me that variance reduction is a fix for non-IID, be careful with this, I think the non-IID claims have to be mostly empirical} The detailed algorithm is outlined in Algorithm \ref{algo:Share}.
%\sk{consider cutting down text somewhere and adding the figure back. Figures are usually very helpful, and some reviewers never bother with the appendix.}

\begin{algorithm}
\caption{SHARE (Secure Hierarchical Robust Aggregation)}\label{algo:Share}
\begin{algorithmic}[1]
%\quad \, $\eta \gets$ global server learning rate \\
%\ENSURE Trained model parameters.\\
%\vspace{0.1cm}
\item\underline{\textbf{Server:}}\\
\FOR{$t = 0, \; \dots \; , T-1$}
    %\STATE Sample subset $\mathcal{S}$ i.e. $f \cdot |\mathcal{S'}|$ clients\\
\FOR{$r=1,\;\dots\;,R$}
\STATE Assign clients to clusters $\mathcal{S}=\M_1\cup\;\dots\; \M_i\;\dots\;\cup \M_c$ with $|\M_i|=|\M_j|\forall i,j\in[c]$
\STATE Compute secure average $g_j^r\gets\texttt{SecureAggr}(\{\Delta_i\}_{i\in \M_j})=\sum_{i\in\mathcal{M}_j} u_i,\forall j\in[c]$
    %\STATE $g^r\gets $ \textsf{Defend}$(g_i) \;\; \forall i \in [m]$
	\STATE $g^r\gets\texttt{RobustAggr}(\{g_j^r\}_{j\in[c]})$
\ENDFOR
\IF {stopping criteria met}
        \STATE break
    \ENDIF
\STATE Push $x^t=x^{t-1}+\eta\frac{1}{R}\sum_rg^r$ to the clients    
\ENDFOR\\
\underline{\textbf{Client:}}
    \FOR{each client $i \in \mathcal{S}$ (if honest) \textbf{in parallel}}
        \STATE $x^{t}_{i, 0} \gets x^t$
        \FOR{$k = 0, \; \dots \; , K-1$}
            \STATE Compute an unbiased estimate $g^{t}_{i,k}$ of $\nabla f_i(x^{t}_{i,k})$
            \STATE $x^{t}_{i,k+1} \gets$ \texttt{ClientOptimize}$(x^{t}_{i,k}, g^{t}_{i,k}, \eta_{l}, k)$
        \ENDFOR
        \STATE $\Delta_i =\frac{n_i}{n}( x^{t}_{i,K} - x^t)$
        \STATE Push $\Delta_i$ to the assigned clusters using secure aggregation\\ 
    \ENDFOR\\
\RETURN  $x^T$
\end{algorithmic} 
\end{algorithm}
\subsection{System Components}
\paragraph{Secure Aggregation:} This is the first step in hierarchical aggregation. We follow an approach similar to \cite{bonawitz2017practical}, using pairwise keys between clients in a cluster. The server in this setup learns just the mean and hence the privacy of individual client updates are protected (Detailed discussion in Appendix C).
\paragraph{Robust Aggregation:} This is the second step in every reclustering round. In this step, the secure cluster averages are filtered through robust aggregation. The goal ideally is to eliminate clusters with malicious client updates. Any existing robustness techniques like trimmed mean\cite{xie2019slsgd}, median\cite{pillutla2019robust} or Zeno\cite{xie2020zeno++} can be utilized at this stage. We show theoretical guarantees and experiments based on existing methods in the following sections.  
\paragraph{Random Reclustering:} As specified in Algorithm \ref{algo:Share}, we repeat the secure aggregation followed by robust aggregation multiple times randomizing client clusters in each global epoch. Note that across these reclustering rounds, the same local model update is paired with different clients each time. In addition to malicious updates, benign updates paired with malicious clients might be filtered in the proposed approach. Reclustering helps mitigate this loss of signal and hence reduces variance. In particular, as number of reclustering rounds ($R$) increase, the probability of this loss in signal decreases (Detailed discussion in Appendix E). 
%This follows naturally from the intuition that it would be harder for the attacker to track the client across all random reclustering rounds.
\begin{remark} Although reclustering increases communication cost, we note that in addition to helping decrease the variance, reducing secure aggregation to within clusters, decreases communication cost as pairwise key exchange is now limited to within the cluster. Hence overall, communication cost for each client changes from $\mathcal{O}(n)$ to $\mathcal{O}(\frac{Rn}{m})$. In experiments, we often find that even a single clustering round gives good results (Section \ref{Sec: Experiments}). 
\end{remark}
\section{Theory}\label{Sec: Theory}
\subsection{Exactness} 
%\sk{edited it}
Algorithm \ref{algo:Share} can be implemented using any aggregation technique. However, due to clustering, the result is resilient to fewer malicious clients -- as (in the worst case) malicious clients are assumed to completely corrupt their assigned cluster. We formalize these ideas next, with proofs in Appendix B.
\begin{lemma}
If robust aggregation is replaced by averaging, the output of Algorithm \ref{algo:Share} is identical to Federated Averaging\cite{mcmahan2017communication}.
\end{lemma}
\begin{lemma}\label{lem: robustaggr}
In presence of robust aggregation, Algorithm \ref{algo:Share} is robust to $q=\frac{q_0}{m}$ adversaries, where $q_0$ is the tolerance limit of the robust aggregation oracle followed and $m$ is the cluster size.
\end{lemma}
\subsection{Convergence Analysis}
 To highlight the flexibility of the proposed algorithm, we analyze convergence when using both using a distance based robust aggregation strategy or a validation data based aggregation strategy, such as Zeno++.
 We first define the few terms used to develop convergence analysis.
\begin{definition}[(G,B)-Bounded Gradient Dissimilarity]\label{bgd}
There exists constants $G\geq0,B\geq1$ such that $\frac{1}{n}\sum_i^n\|\nabla f_i(x)\|^2\leq G^2+B^2\|\nabla F(x)\|^2$
\end{definition}
\begin{definition}[Bounded client updates variance]\label{sigmag}
We define benign mean model update across clients to be $\mu=\sum_i\Delta_i^K$, hence the variance across client updates as $\mathbb{E}[\|\Delta_i^K-\mu\|^2]\leq\sigma_g^2$ for all $i$ across all rounds of training
\end{definition}
\begin{definition}[Bounded variance]\label{sigma}
For an unbiased stochastic gradient estimator with $g_i(x)=\nabla f_i(x,z_i)$ we define bounded variance as $\mathbb{E}_{z_i}[\|g_i(x)-\nabla f_i(x)\|]\leq\sigma^2$ for any $i,x$
\end{definition}
The difference between Definition~\ref{sigmag} and \ref{sigma} is that the former bounds the variance between model updates across clients while the latter bounds the variance across gradient estimates within the same client. 
\subsubsection{Convergence Rates}
 We now prove that Algorithm \ref{algo:Share} converges for various robust aggregation oracles. Firstly, we state a few general assumptions required to prove convergence guarantees standard in papers.
 \begin{assumption}\label{minima}
 There exists at least one global minima $x^*$ such that $F(x^*)\leq F(x),\forall x$
 \end{assumption}
 \begin{assumption}\label{taylor}
 We assume that $F(x)$ is L-smooth and has $\mu$-lower bounded Taylor approximation ($\mu$ weak convexity)
 \begin{align*}
     \langle\nabla F(x),y-x\rangle+\frac{\mu}{2}\|y-x\|^2\leq F(y)-F(x)\leq\langle\nabla F(x),y-x\rangle+\frac{L}{2}\|y-x\|^2
 \end{align*}
 \end{assumption}
 Note that this Assumption \ref{taylor} covers the case of non-convexity by taking $\mu<0$.
 %, non-strong convexity when $\mu=0$ and $\mu$-strong convex otherwise. 
 We note that each distance based robust aggregation metric have different bounds from benign mean update. Since the focus of this work is to propose an algorithm that unifies robustness with privacy, we do not concentrate on those bounds and absorb such intricacies into an order constant. Formally,
 \begin{assumption}\label{robustaggr}
 For any distance based robust aggregation algorithm, \textit{when fraction of faulty inputs is below threshold}, the output of robust aggregation is bounded from benign mean. That is, we assume there exists a $V_2$ such that for any set of vectors $\{v_i: i\in\mathcal{C}\}$, replaced by faulty vectors $\forall i\notin \mathcal{C}_t\subseteq\mathcal{C}$,  $\|\texttt{RobustAggr}(\{v_i\}_{i\in\mathcal{C}})-\frac{1}{|\mathcal{C}_t|}\sum_{i\in\mathcal{C}_t}v_i\|\leq\mathcal{O}( V_2)$.
 \end{assumption}
We note that Assumptions \ref{minima},\ref{taylor} are standard among existing Federated Learning literature~\citep{karimireddy2020scaffold,xie2019zeno,xie2020zeno++}. Additionally Assumption \ref{robustaggr} is a direct consequence of existing distance based robust aggregation oracles~\citep{xie2019slsgd,pillutla2019robust,blanchard2017machine}. Finally, for Algorithm \ref{algo:Share} with such oracles, we have the following theorem
\begin{theorem}\label{Theorem: distaggr}
Consider a function F(x) satisfying Assumptions \ref{minima},\ref{taylor} assume a robust aggregation scheme that picks up $b$ updates and satisfies Assumption \ref{robustaggr}, further, assume (G,B)-Bounded gradient dissimilarity, $\sigma_g^2$ variance in client updates and $\sigma^2$ variance in gradient estimation, there exists $\eta,\eta_l$ such that output of Algorithm \ref{algo:Share} after T rounds, $x^T$, satisfies,
\begin{align*}
    \mathbb{E}\left[\|\nabla F(x^T)\|^2\right]\leq\mathcal{O}\left(\frac{LM\sqrt{F}}{\sqrt{TKn}}+\frac{F^{2/3}(LG^2)^{1/3}}{(T+1)^{2/3}}+\frac{B^2LF}{T}+2L^2V_2+\frac{\sigma_g^2}{bm}\left(\frac{n-q-bm}{R(n-q)-1}\right)\right)
\end{align*}
where $M^2:=\sigma^2(1+\frac{n}{\eta^2})$ and $F:=F(x^0)-F(x^*)$
\end{theorem}
Now we consider Zeno++\cite{xie2020zeno++}, a defense utilizing server data. Although score based Zeno++ was originally introduced for asynchronous SGD, we generalize it to federated learning setting hence allowing for multiple local epochs. We illustrate this modified algorithm in Appendix B. As in \citet{xie2020zeno++}, we consider an additional standard assumption 
\begin{assumption}\label{valassump}
The validation set considered for Zeno++ is close to training set, implying a bounded variance given by $\mathbb{E}[\|\nabla f_s(x)-\nabla F(x)\|^2]\leq V_1,\forall x$
\end{assumption}
\begin{theorem}\label{Theorem: Zeno}
Consider L-smooth and potentially non-convex functions $F(x)$ and $f_s(x)$, satisfying Assumption \ref{valassump}. Assume $\|f_s(x)\|^2\leq V_3,\forall x$. Further assuming G-bounded gradient dissimilarity, variance between client updates be $\sigma_g^2$ and variance in gradient estimation at each client be $\sigma$, with global and local learning rates of $\eta\leq\frac{1}{2L}$ and $\rho\geq\frac{\alpha\sqrt{\eta}}{6K^2\eta_l^2B^2}+\eta$, after T global updates, let $D:=F(x^0)-F(x^*)$, Algorithm \ref{algo:Share} with Zeno++ as robust aggregation converges at a critical point:
\begin{align*}
    \frac{\mathbb{E}[\sum_{t\in[T]}\|\nabla F(x_{t-1})\|^2]}{T}\leq\frac{\mathbb{E}[D]}{\alpha\sqrt{\eta}T}+\frac{\sqrt{\eta}}{\alpha}\mathcal{O}\left(\frac{\sigma_g^2}{m}\left(\frac{n-q-m}{R(n-q)-1}\right)+G^2+\sigma^2+V_1+V_3+\epsilon\right)
\end{align*}
\end{theorem}
%\sk{Fix notation e.g., overusing F as both loss, and difference of losses. Also the theorem need some smoothing. Maybe pull out the validation-train gap as an assumption?}
\begin{remark}
It can be seen from both Theorem \ref{Theorem: distaggr},\ref{Theorem: Zeno} that the additional terms, other than standard ones appearing in the convergence rate for federated learning~\citep{karimireddy2020scaffold}, depend on the error caused by the robust aggregation scheme utilized and variance reduction from reclustering. Further, higher number of reclustering rounds $R$ decreases the effect of additional variance. Finally when $R=1,m=1,q=0$, these recover existing results for federated learning with robust aggregation.
\end{remark}
%\begin{remark}
%As with distance based aggregations, additional terms corresponding to client clustering appear in the rates although this is scaled down by reclustering. As reclustering rounds increase, this additional variance decreases.  
%\end{remark}
\subsection{Privacy}
\paragraph{Curious server:} Since each client masks updates with random vectors as illustrated in Section \ref{Sec: Solution}, we note that if we execute the mentioned secure averaging oracle with threshold $t>\frac{m}{2}$, the protocol can deal with $\lceil\frac{m}{2}\rceil-1$ drop outs while learning nothing more than average. Reclustering introduces additional vulnerability as server can see multiple averages. In particular, the probability that server can decode a model update is $\mathcal{O}(1-\left(\frac{(m!)^c}{n!}\right)^R)$. Hence as R increases this gets closer to 1 as expected. Further, when all clients are in a single cluster ($m=n$, hence c=1), this is 0 as would be the case with secure averaging without robustness. Further discussion, including comments on privacy in presence of colluding curious clients can be found in Appendix C.
\section{Experiments}\label{Sec: Experiments}
In this section we evaluate the proposed algorithm SHARE with various defenses and corruption models. We conduct experiments on CIFAR-10~\cite{krizhevsky2009learning}~(Image classification dataset) and Shakespear (a language modeling dataset from LEAF~\cite{caldas2018leaf}). We note that we do not propose a new robustness technique but rather we propose a modified federated learning architecture to incorporate any robustness protocol in a privacy preserving manner. Hence we focus our experiments on capturing the effects of cluster sizes and reclustering rounds, hyperparameters introduced by our approach. We defer descriptions of detailed training architecture to Appendix D
\subsection{CIFAR-10}
 We train a CNN with two $5\times5$ convolutional layers followed by 2 fully connected layers\cite{mcmahan2017communication} on CIFAR-10 and report top-1 accuracy. We test SHARE incorporating various robust aggregation protocols such as Trimmed mean~\cite{xie2019slsgd}, Krum~\cite{blanchard2017machine}, Zeno++~\cite{xie2020zeno++}. For all experiments in this section, trimmed mean removes $2/3$ of the updates before computing the mean. Additionally, we consider two baselines, SHARE with no robust aggregation and SHARE with no attack. We consider homogeneous distribution of data across clients for experiments in this section. Experiments on heterogeneous data distributions can be found in Appendix D.
\subsubsection{Impact of cluster size}\label{clustersize}
We first test Byzantine-tolerance for various cluster sizes to mild attacks such as label-flip. In particular, malicious clients train on wrong labels (images whose labels are flipped, i.e., any label $\in\{0,\hdots,9\}$ is changed to 9-label). We consider 60 total clients of which $q=12$ being malicious.
\begin{figure}[h]
\centering
\subcaptionbox{Zeno++, q=12\label{zeno-lf}}
    {\includegraphics[width=0.32\textwidth]{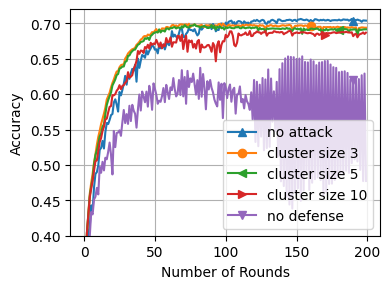}}
\subcaptionbox{Krum, q=12\label{krum-lf}}
    {\includegraphics[width=0.32\textwidth]{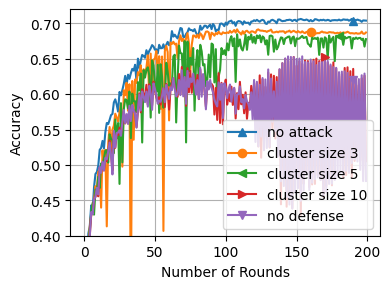}}
\subcaptionbox{Trimmed mean, q=12\label{trimmed-mean-lf}}
    {\includegraphics[width=0.32\textwidth]{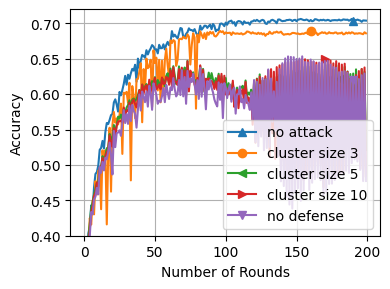}}
\caption{Results of SHARE with various defenses on CIFAR-10, utilizing varying cluster sizes under label flip attack. For Trimmed mean we remove 2/3 of input updates. %and use batch size of 128, $\rho=0.0001,\epsilon=0.2$ as Zeno++ parameters.
\label{fig: labelflip}}
\end{figure}
The result is shown in Figure \ref{fig: labelflip} for various cluster sizes and robust aggregation protocols. It is seen that having no defense diverges even with mild attacks as expected. Further Figure (\ref{zeno-lf}) shows that SHARE with a strong defense like Zeno++ converges to benign (no-attack) accuracy for any of the considered cluster sizes. SHARE with trimmed mean and Krum both converge with cluster size 3 but as cluster size increases, accuracy decreases and SHARE begins to diverge. This can be seen directly from Lemma \ref{lem: robustaggr}, since we set trimmed mean to filter $q_0=(2/3)*60=40$ of updates, a cluster size of 3 implies the algorithm is robust against $q=40/3>12$ clients being malicious, hence the algorithm converges to benign accuracy, increasing the cluster size decreases this tolerance threshold and hence as shown in Figure (\ref{trimmed-mean-lf}) may fail to converge. Further experiments on scaled sign-flip attacks, are included in Appendix D due to space constraints.

% Further, we now test the Byzantine-tolerance to a stronger attack sign-flip attack where the malicious gradients report an update scaled by $-10$, hence both changing the direction, scaling the update to magnify the effect. Since Zeno++ can tolerate a greater fraction of clients being Byzantine, we set $q=6$ while for trimmed mean and Krum, we set $q=3$. 
% \begin{figure}[h]
% \centering
% \subcaptionbox{Zeno++, $-10\Delta$\label{zeno-10g}}
%     {\includegraphics[width=0.32\textwidth]{zeno-10g}}
% \subcaptionbox{Krum, $-10\Delta$\label{krum-10g}}
%     {\includegraphics[width=0.32\textwidth]{krum-10g}}
% \subcaptionbox{Trimmed mean, $-10\Delta$\label{trimmed-mean-10g}}
%     {\includegraphics[width=0.32\textwidth]{trimmed-mean-10g}}
% \caption{SHARE with sign-flipping attack on CIFAR-10 \sk{Why are the "no defense" results different for zeno vs others?; Also, often good to add details in the caption so a user can just read and get the gist}}\label{fig: -10g}
% \end{figure}
% The results are plotted in Figure \ref{fig: -10g}. A similar effect of cluster sizes can be seen here with Zeno++ being stable even at higher levels of corruption. 
\subsubsection{Impact of reclustering}
Intuitively, increasing the number of reclustering rounds increases the expected number of clusters without a Byzantine client. This hence increases the robustness of SHARE to higher fraction of Byzantine clients with defenses like Zeno++ which can tolerate arbitrary levels of poisoning. We test this hypothesis with several attack and clustering scenarios using sign-flipping attacks. 
\begin{figure}[h]
\centering
\subcaptionbox{$q=6$, $m=10$, -10$\Delta$\label{resample-normal}}
    {\includegraphics[width=0.32\textwidth]{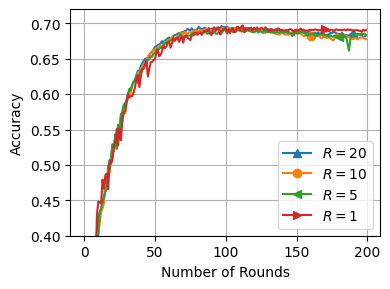}}
\subcaptionbox{$q=30$, $m=4$, -10$\Delta$\label{resample-byz}}
    {\includegraphics[width=0.32\textwidth]{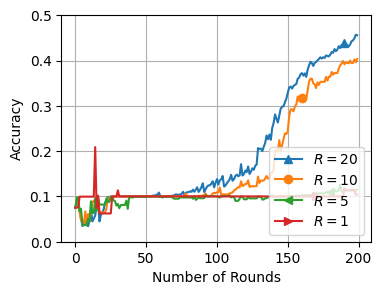}}
\subcaptionbox{$q=6$, $m=25$, -100$\Delta$ \label{resample-group}}
    {\includegraphics[width=0.32\textwidth]{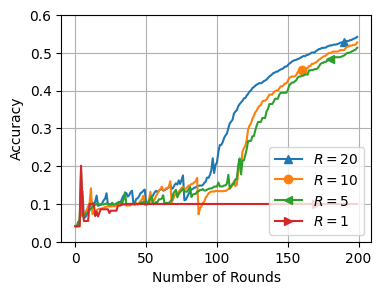}}
\caption{Results of SHARE with Zeno++%(batch size of 128, $\rho=0.0001,\epsilon=0.2$) 
and different reclustering rounds $R$, Byzantine clients $q$, cluster sizes $m$ on CIFAR-10 with varying attack strengths (Any benign model update $\Delta$ is scaled to either $-10\Delta$ or $-100\Delta$). In (a),(b) we use $n=60$ and for (c) we use $n=100$.}
\end{figure}
In (a), we use a relatively small cluster size and a low fraction of Byzantine clients, so 1 round is sufficient. In (b), the fraction of Byzantine clients is high and in (c) the cluster size is large, which increases the probability of a cluster containing a Byzantine client, so $R>1$ helps converge to higher accuracies.
\subsection{Shakespeare}
%Shakespeare~\cite{caldas2018leaf} is a language modeling dataset built on collective works of William Shakespeare. 
%We construct train and test sets using the LEAF framework. %where the train set at each client corresponds to a different speaking role. 
We consider the first 60 speaking roles in the train set as our 60 clients. We train an RNN with 2 LSTM layers followed by 1 fully connected layer\cite{reddi2020adaptive} and report top-1 accuracy on the testing set.
\subsubsection{Empirical Evaluation}
In Figure \ref{fig: shake} we evaluate Byzantine tolerance of SHARE with Zeno++ under sign-flip attack (malicious clients send an update negative to the benign one $-\Delta$) and scaled sign-flip attack (malicious clients scale the update in addition to flipping the sign and hence send $-10\Delta$). A stronger attack like scaled sign-flip breaks benign averaging and Zeno++ works well with any of the chosen cluster sizes.
\begin{figure}[h]
\centering
\subcaptionbox{Zeno++, $-10\Delta$\label{shakespeare-10g}}
    {\includegraphics[width=0.32\textwidth]{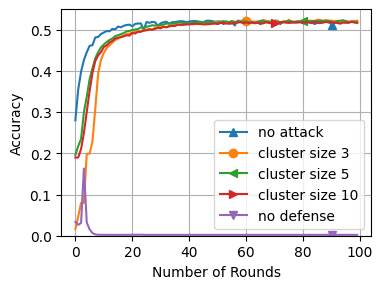}}
\subcaptionbox{Zeno++, $-\Delta$\label{shakespeare-g}}
    {\includegraphics[width=0.32\textwidth]{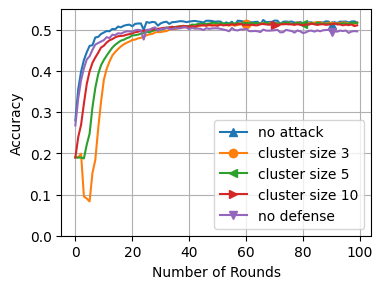}}
\caption{SHARE with Zeno++ defense and sign-flipping attack on Shakespeare.}\label{fig: shake}
\end{figure}
\section{Discussion and Conclusion}
We have proposed SHARE, a framework for implementing Byzantine-robustness and privacy. The key idea is hierarchical clustering. Cluster size is an important parameter that controls the trade-off between privacy and robustness. %An extreme case example of cluster size is all clients, which ignores robust aggregation and recovers vanilla Federated Averaging, while a cluster size of 1 recovers Byzantine-robust aggregation without privacy. 
Further, reclustering is an important step and can help decrease variance and increase tolerance to the fraction of malicious clients when the defense can support arbitrary failures like Zeno++. In future, we would like to explore other variations in client clustering, especially in heterogeneous data settings. Further, we plan to work on stronger security guarantees even with multiple reclustering rounds.

\bibliography{refs}

\appendix

%\section{Appendix}

\newpage
 \appendix

 \section*{Appendix}
% \section{Illustration}
% Figure \ref{fig:shareillus} illustrates Algorithm \ref{algo:Share} (SHARE).
%\begin{figure}[h]
%\centering
%\includegraphics[scale=0.5]{illustration_image3.png}
%\caption{This figure illustrates our proposed framework SHARE. Each global round consists of multiple reclustering rounds, updates from which are averaged to obtain the final model update. In each reclustering round (shown by dotted rectangle), updates from clients (numbered squares) are clustered randomly ($\mathcal{M}_1,\mathcal{M}_2,\mathcal{M}_3$), then averaged. These aggregate updates (dotted lines) are combined by the server.}
%, which are averaged within each cluster to reveal the client clustered updates for robust aggregation at the server.}
%\label{fig:shareillus}
%\end{figure} 
\section{Notations}  
\begin{table}[h]
\begin{tabular}{ | m{1.2cm}| m{13em}| m{1.4cm}|m{14em}|} 
 \hline
 Notation & Description & Notation & Description\\ \hline
 $n$ & Total number of clients &$c$ & Number of clusters \\
 $q$ & Number of faulty clients&$\mathcal{S}$ &set of all clients\\
 $K$ & Number of local SGD epochs  &$[m]$& The set of integers $\{1,\hdots,m\}$\\
 $T$ & Number of global epochs &$\{\M_j\}_{j\in[c]}$ &Set of client clusters\\
 $R$ & Number of resampling epochs &
 $n_i$ & Number of samples on worker $i$\\
 $b$ & Trim parameter for defense& $m$ & Number of clients in each cluster \\
 $\eta_l,\eta$ & Local and global learning rates &
 $\|\cdot\|$ & All norms in this paper are $l_2$-norms\\
 \hline
\end{tabular}\caption{\label{tab: notations} Notations utilized in this paper}
\end{table}
 \section{Proofs}
 In this section, we elaborate on theoretical guarantees of SHARE. We define the following quantity to aid the proofs that follow
  \begin{definition}[Clustered Client Update]
We define clustered client update as average model updates from all the clients assigned to a particular cluster. Mathematically, the clustered client update in a reclustering round $r\leq R$ is  given by $g_i^r=\sum_{j\in \mathcal{M}_k}\Delta_j$ where $\Delta_j$ denotes the model update from client $j$ belonging to cluster $\mathcal{M}_k$ after $K$ steps of SGD. 
\end{definition}
\begin{lemma}
If robust aggregation is replaced by averaging, output of Algorithm \ref{algo:Share} is identical to Federated Averaging\cite{mcmahan2017communication}.
\end{lemma}
\begin{proof}
In each re-clustering round, the update with benign averaging becomes $g^r=\sum_{i\in[c]}\sum_{j\in \M_i}\Delta_j^K=\sum_{l\in[n]}\Delta_l^K$ where $n$ is the total number of clients and $c=\frac{n}{m}$ is the total number of clusters, as this update is independent of the random cluster division, the global update at round $t$ becomes $x^t=x^{t-1}+\eta\sum_{l\in[n]}\Delta_l^K$ which is identical to federated averaging.
\end{proof} 
\begin{lemma}%\label{lem: robustaggr}
In presence of robust aggregation, Algorithm \ref{algo:Share} is robust to $q=\frac{q_0}{m}$ where $q_0$ is the tolerance limit of the robust aggregation oracle followed and $m$ is the cluster size.
\end{lemma}
\begin{proof}
We consider the worst case scenario of each malicious client being in different clusters, hence spreading the attack to the maximum possible number of clients. Although randomization beats this and might offer better clusters in multiple random rounds, there might still exist such attack favorable rounds. Allowing for this worst case sets the threshold to $q=\frac{q_0}{m}$ if the original robustness oracle has a threshold of $q_0$.
\end{proof}

 \subsection{Distance based robust aggregation}
 \begin{theorem}
Consider a function F(x) satisfying Assumptions \ref{minima},\ref{taylor} assume a robust aggregation scheme that picks up $b$ updates and satisfies Assumption \ref{robustaggr}, further, assume (G,B)-Bounded gradient dissimilarity, $\sigma_g^2$ variance in client updates and $\sigma^2$ variance in gradient estimation, there exists $\eta,\eta_l$ such that output of Algorithm \ref{algo:Share} after T rounds, $x^T$, satisfies,
\begin{align*}
    \mathbb{E}\left[\|\nabla F(x^T)\|^2\right]\leq\mathcal{O}\left(\frac{LM\sqrt{F}}{\sqrt{TKn}}+\frac{F^{2/3}(LG^2)^{1/3}}{(T+1)^{2/3}}+\frac{B^2LF}{T}+2L^2V_2+\frac{\sigma_g^2}{bm}\left(\frac{n-q-bm}{R(n-q)-1}\right)\right)
\end{align*}
where $M^2:=\sigma^2(1+\frac{n}{\eta^2})$ and $F:=F(x^0)-F(x^*)$
\end{theorem}
 \begin{proof} Firstly, we bound the distance between global model update from Algorithm \ref{algo:Share} and  expected benign mean model update in each global iteration. In particular, let the expected benign mean model update be denoted by $\mu_t$ and global model update in each iteration is given by $\frac{1}{R}\sum_r\texttt{RobustAggr}(\{g_i^r\}_{i\in[c]})$. We determine an upper bound on $\|\mathbb{E}[\frac{1}{R}\sum_r\texttt{RobustAggr}(\{g_i^r\}_{i\in[c]})]-\mu_t\|$. This is illustrated below
 \begin{align*}
     \|\mathbb{E}[\frac{1}{R}\sum_r\texttt{RobustAggr}(\{g_i^r\}_{i\in[c]})]-\mu_t\|^2&\leq\|\mathbb{E}[\frac{1}{R}(\sum_r\texttt{RobustAggr}(\{g_i^r\}_{i\in[m]})-\sum_{r,i\in B}g_i^r)]\\&+\mathbb{E}[\frac{1}{R}\sum_{r,i\in B}g_i^r]-\mu_t\|^2\\
     &\leq\mathcal{O}(V_2)+2\|\mathbb{E}[\frac{1}{R}\sum_{r,i\in B}g_i^r]-\mu_t\|^2\\
     &\leq\mathcal{O}(V_2)+2\mathbb{E}\|\texttt{Resample}(\{\Delta_i^K\}_{i\in C})-\mu_t\|^2\\
     &\leq\mathcal{O}(V_2)+\frac{\sigma_g^2}{Rbm}(1-\frac{R(bm)-1}{R(n-q)-1})\\
     &\leq\mathcal{O}(V_2)+\frac{\sigma_g^2}{bm}(\frac{n-q-bm}{R(n-q)-1})
 \end{align*}
 Where $B$ denotes indices of benign clusters (clusters with uncorrupted device updates. Mathematically, let $\mathcal{C}_t$ denote set of benign clients among all $n$ clients. $\{B:i\in [c]\text{ such that } \forall k\in[m],\Delta_k\in\mathcal{M}_i,\Delta_k\in\mathcal{C}_t\}$ ), $r\leq R$ as mentioned in the text denote reclustering rounds. The second inequality follows from Assumption \ref{robustaggr}. Since each reclustering round randomly groups clients together, the set $\{\Delta_k: \Delta_k\in\mathcal{M}_i,i\in B\}$ is a random resample of $bm$ benign client updates from the available $n-q$, where $b$ is the number of updates available after filteration through robust aggregation. With $R$ resampling rounds, this is equivalent to resampling $Rbm$ updates from $R(n-q)$ benign updates. Following \citet{rice2006mathematical}(Chapter 7, Theorem B), we obtain the scaled down variance bound. 
 
 Using L-smoothness of $F(x)$,
 \begin{align*}
     \mathbb{E}[\|\nabla F(x^t)\|^2]&\leq2\mathbb{E}[\|\nabla F(x^t)-\nabla F(\mu_t)\|^2]+2\mathbb{E}[\|F(\mu_t)\|^2]\\
     &\leq2L^2(\mathcal{O}(V_2)+\frac{\sigma_g^2}{bm}(\frac{n-q-bm}{R(n-q)-1}))+2\mathbb{E}[\|F(\mu_t)\|^2]
 \end{align*}
 The rest follows a similar approach as \citet{karimireddy2020scaffold} hence we get
 \begin{align*}
    \mathbb{E}\left[\|\nabla F(x^T)\|^2\right]\leq\mathcal{O}\left(\frac{LM\sqrt{F}}{\sqrt{TKn}}+\frac{F^{2/3}(LG^2)^{1/3}}{(T+1)^{2/3}}+\frac{B^2LF}{T}+2L^2V_2+\frac{\sigma_g^2}{bm}\left(\frac{n-q-bm}{R(n-q)-1}\right)\right)
\end{align*}
where $M^2:=\sigma^2(1+\frac{n}{\eta^2})$ and $F:=F(x^0)-F(x^*)$.
 \end{proof}

 \subsection{Zeno++ as robust aggregation}
 We first illustrate a modified Zeno++ algorithm and adapt it to Federated Learning setting from its original asynchronous SGD paradigm. Firstly, we define a score that helps filter out updates if they fall below a threshold. Intuitively, the score denotes trustworthiness of a clustered update.
 \begin{definition}[Approximated model update score] Denote $f_s(x)=\frac{1}{n_s}\sum_i^{n_s}f(x;z_i)$, where $z_j's$ are drawn independent and identically from $\mathcal{D}_s\neq\mathcal{D}_i,\forall i\in[n]$ and $n_s$ is the batch size of $f_s(\cdot)$, for a clustered client update $g$, model parameter $x$, global learning rate $\eta$ and constant weight $\rho>0$, we define model update score as
 \begin{align*}
     Score_{\eta,\rho}\approx-\eta\langle\nabla f_s(x),g\rangle-\rho\|g\|^2
 \end{align*}
 where $x$ is the current model available on the server.
 \end{definition}
Using this approximated model update score, we set hard thresholding parameterized by $\epsilon$ to filter client cluster updates. Algorithm \ref{algo:ZenoShare} illustrates SHARE framework with Zeno++ as robust aggregation. We analyze the convergence of Algorithm \ref{algo:ZenoShare} in the following theorem.
 \begin{algorithm}
\caption{SHARE (Secure Hierarchical Robust Aggregation) with Zeno++ defense}\label{algo:ZenoShare}
\begin{algorithmic}[1]
%\quad \, $\eta \gets$ global server learning rate \\
%\ENSURE Trained model parameters.\\
%\vspace{0.1cm}
\item\underline{\textbf{Server:}}\\
\FOR{$t = 0, \; \dots \; , T-1$}
    %\STATE Sample subset $\mathcal{S}$ i.e. $f \cdot |\mathcal{S'}|$ clients\\
\FOR{$r=1,\;\dots\;,R$}
\STATE Assign clients to clusters $\mathcal{S}=\M_1\cup\;\dots\; \M_i\;\dots\;\cup \M_c$ with $|\M_i|=|\M_j|\forall i,j\in[c]$
\STATE Compute secure average $g_j^r\gets\texttt{SecureAggr}(\{\Delta_i\}_{i\in \M_j})=\sum_{i\in\mathcal{M}_j} u_i,\forall j\in[c]$
    %\STATE $g^r\gets $ \textsf{Defend}$(g_i) \;\; \forall i \in [m]$
    \STATE Randomly sample $z_j\sim S,\forall j\in [n_s]$ to compute $f_s$
     \FOR{$i=1,\;\dots\;,c$}
     \IF{score($g_i^r,x^{t-1}\geq-\eta\epsilon$)}
         \STATE $g^r\gets g^r+g_i^r$,
     \ENDIF
     \ENDFOR
	%\STATE $g^r\gets\texttt{RobustAggr}(\{g_j^r\}_{j\in[c]})$
\ENDFOR
     
\IF {stopping criteria met}
        \STATE break
    \ENDIF
\STATE Push $x^t=x^{t-1}+\eta\frac{1}{R}\sum_rg^r$ to the clients    
\ENDFOR\\
\underline{\textbf{Client:}}
    \FOR{each client $i \in \mathcal{S}$ (if honest) \textbf{in parallel}}
        \STATE $x^{t}_{i, 0} \gets x^t$
        \FOR{$k = 0, \; \dots \; , K-1$}
            \STATE Compute an unbiased estimate $g^{t}_{i,k}$ of $\nabla f_i(x^{t}_{i,k})$
            \STATE $x^{t}_{i,k+1} \gets$ \texttt{ClientOptimize}$(x^{t}_{i,k}, g^{t}_{i,k}, \eta_{l}, k)$
        \ENDFOR
        \STATE $\Delta_i =\frac{n_i}{n}( x^{t}_{i,K} - x^t)$
        \STATE Push $\Delta_i$ to the assigned clusters using secure aggregation\\ 
    \ENDFOR\\
\RETURN  $x^T$
\end{algorithmic} 
\end{algorithm}
\begin{theorem}
Consider L-smooth and potentially non-convex functions $F(x)$ and $f_s(x)$, Assume validation set is close to training set, implying a bounded variance given by $\mathbb{E}[\|\nabla f_s(x)-\nabla F(x)\|^2]\leq V_1,\forall x$, Assume $\|f_s(x)\|^2\leq V_3,\forall x$. Further assuming bounded gradient dissimilarity as stated in \ref{bgd},variance between client updates of $\sigma_g^2$ and variance in gradient estimation at each client be $\sigma$, with global and local learning rates of $\eta\leq\frac{1}{2L}$ and $\rho\geq\frac{\alpha\sqrt{\eta}}{6K^2\eta_l^2B^2}+\eta$, after T global updates, let $F:=F(x^0)-F(x^*)$, Algorithm \ref{algo:ZenoShare} with Zeno++ as robust aggregation converges at a critical point:
\begin{align*}
    \frac{\mathbb{E}[\sum_{t\in[T]}\|\nabla F(x_{t-1})\|^2]}{T}\leq\frac{\mathbb{E}[F]}{\alpha\sqrt{\eta}T}+\frac{\sqrt{\eta}}{\alpha}\mathcal{O}\left(\frac{\sigma_g^2}{m}\left(\frac{n-q-m}{R(n-q)-1}\right)+G^2+\sigma^2+V_1+V_3+\epsilon\right)
\end{align*}
\end{theorem}
 \begin{proof}
 Since for any cluster update $g$ that passes the test of Zeno++, it follows that
 \begin{align*}
     -\langle\nabla f_s(x_t),\eta g\rangle-\rho\|g\|^2\geq-\eta\epsilon.
 \end{align*}
 Thus, we have
 \begin{align*}
     &\langle \nabla F(x_{t-1}),\eta\mathbb{E}[g^r]\rangle\\
     &\leq\langle \nabla F(x_{t-1})-\nabla f_s(x_t),\eta\mathbb{E} [g^r]\rangle-\rho\mathbb{E}[\|g^r\|^2]+\eta\epsilon\\
     &\leq\frac{\eta}{2}\|\nabla F(x_{t-1})-\nabla f_s(x_t)\|^2+\frac{\eta}{2}\|\mathbb{E}[g^r]\|^2-\rho\mathbb{E}[\|g^r\|^2]+\eta\epsilon\\
     &\leq\frac{\eta}{2}\|\nabla F(x_{t-1})-\nabla f_s(x_t)\|^2+(\frac{\eta}{2}-\rho)\mathbb{E}[\|g^r\|^2]+\eta\epsilon\\
     &\leq\frac{\eta}{2}\|\nabla F(x_{t-1})-\nabla F(x_t)+\nabla F(x_t)-\nabla f_s(x_t)\|^2+(\frac{\eta}{2}-\rho)\mathbb{E}[\|g^r\|^2]+\eta\epsilon\\
     &\leq\eta\|\nabla F(x_{t-1})-\nabla F(x_t)\|^2+\eta\|\nabla F(x_t)-\nabla f_s(x_t)\|^2+(\frac{\eta}{2}-\rho)\mathbb{E}[\|g^r\|^2]+\eta\epsilon\\
     &\leq\eta\|\nabla F(x_{t-1})-\nabla F(x_t)\|^2+\eta V_1+(\frac{\eta}{2}-\rho)\mathbb{E}[\|g^r\|^2]+\eta\epsilon\\
     &\leq\eta^3L^2\mathbb{E}[\|g^r\|^2]+\eta V_1+(\frac{\eta}{2}-\rho)\mathbb{E}[\|g^r\|^2]+\eta\epsilon\\
 \end{align*}
 Where $g^r$ is the model update in reclustering round $r\leq R$. From $L$ smoothness, we have
 \begin{align*}
     \|\nabla F(x_{t-1})-\nabla F(x_t)\|^2\leq L^2\|x_{t-1}-x_t\|^2\leq L^2\eta^2\mathbb{E}[\|g^r\|^2]
 \end{align*}
 Using smoothness again, considering a global step size as $\eta L\leq \frac{1}{2}$we get
 \begin{align}
     \mathbb{E}&[F(x_t)]\\
     &\leq F(x_{t-1})+\langle\nabla F(x_{t-1}),\eta\mathbb{E}[g^r] \rangle+\frac{L\eta^2}{2}\mathbb{E}[\|g^r\|^2]\\
     &\leq F(x_{t-1})+(\eta^3L^2+\frac{\eta}{2}-\rho+\frac{L\eta^2}{2})\mathbb{E}[\|g^r\|^2]+\eta V_1+\eta\epsilon\\
     &\leq F(x_{t-1})+(\eta-\rho)\mathbb{E}[\|g^r\|^2]+\eta V_1+\eta\epsilon\label{eq: halfderiv}
 \end{align}
 Now we will bound the term $\mathbb{E}[\|g^r\|^2]$. Further, $\mathbb{E}[\|g^r\|^2]\leq 2(\frac{V_3}{2}+\eta\epsilon)+\mathbb{E}\|\tilde{g}^r\|^2$ where $\|\nabla f_s(x)\|^2\leq V_3$ and $\tilde{g}^r$ is benign average obtained through sampling of benign clients
 \begin{align}
     \mathbb{E}\|\tilde{g}^r\|^2\leq\mathbb{E}\|x_i^K-x_{t-1}\|^2+\frac{\sigma_g^2}{m}(\frac{n-q-m}{R(n-q)-1})\label{eq: boundingg}
 \end{align}
 Where $x_i^K$ corresponds to model parameters after K rounds of SGD on $ith$ device, $\sigma_g^2$ corresponds to variance between device updates. Since sampling successive $g_i's$ can be seen as sampling with replacement, and at least one cluster is selected each time, this has a maximum variance of single cluster selection case. ($m$ is the cluster size, $n$ is the total number of devices). The mean is equal to client drift, which can be bounded as shown below (for notational brevity, present global model $x_{t-1}$ is denoted as $x$). Let us assume gradients at each data point $g_i(x_i^{k-1})=\nabla f_i(x_i^{k-1})+\text{error}$, where error has mean 0 and $\sigma$ standard deviation as stated in Assumption \ref{sigma}. For $k\leq K$ steps of local SGD, we get
 \begin{align*}
     \mathbb{E}\|x_i^k-x\|^2&=\mathbb{E}\|x_i^{k-1}-x-\eta_lg_i(x_i^{k-1})\|^2\\
     &\leq\mathbb{E}\|x_i^{k-1}-x-\eta_l\nabla f_i(x_i^{k-1})\|^2+\eta_l^2\sigma^2\\
     &\leq(1+\frac{1}{K-1})\mathbb{E}\|x_i^{k-1}-x\|^2+K\eta_l^2\|\nabla f_i(x_i^{k-1})\|^2+\eta_l^2\sigma^2\\
     &\leq(1+\frac{1}{K-1})\mathbb{E}\|x_i^{k-1}-x\|^2+2K\eta_l^2\|\nabla f_i(x_i^{k-1})-\nabla f_i(x)\|^2+2K\eta_l^2\|\nabla f_i(x)\|^2+\eta_l^2\sigma^2\\
     &\leq(1+\frac{1}{K-1}+2K\eta_l^2L^2)\mathbb{E}\|x_i^{k-1}-x\|^2+2K\eta_l^2\|\nabla f_i(x)\|^2+\eta_l^2\sigma^2
 \end{align*}
 Where the first inequality uses mean and variance in gradient estimation, the second one follows from relaxed triangle inequality as stated in \citet{karimireddy2020scaffold}(Lemma 3). Taking appropriate local step size $\eta_l^2\leq\frac{1}{2L^2K(K-1)}$ and telescoping the sum, we get
 \begin{align*}
     \mathbb{E}\|x_i^k-x\|^2&\leq\sum_{\tau=1}^{k-1}(2K\eta^2_l\|\nabla f_i(x)\|^2+\eta^2_l\sigma^2)(1+\frac{2}{K-1})^\tau\\
     &\leq(2K\eta^2_l\|\nabla f_i(x)\|^2+\eta^2_l\sigma^2)\sum_\tau(1+\frac{2}{K-1})^\tau\\
     &\leq(2K\eta^2_l\|\nabla f_i(x)\|^2+\eta^2_l\sigma^2)3K
 \end{align*}
 The last inequality follows from the fact that $\tau<K$ and $(1+x/n)^n<\exp(x)$. Substituting this back into (\ref{eq: boundingg}) and averaging over all $i's$ (client devices), we get
 \begin{align*}
     \mathbb{E}\|g^r\|^2&\leq\frac{1}{N}6K^2\eta^2_l\sum_i\|\nabla f_i(x)\|^2+3K\eta^2_l\sigma^2+\frac{\sigma_g^2}{m}(\frac{n-q-m}{R(n-q)-1})+V_3+2\eta\epsilon\\
     &\leq 6K^2\eta_l^2G^2+6K^2\eta_l^2B^2\|\nabla F(x)\|^2+3K\eta_l^2\sigma^2+\frac{\sigma_g^2}{m}(\frac{n-q-m}{R(n-q)-1})+V_3+2\epsilon\eta\\
 \end{align*}
 Where $\frac{1}{N}\sum_i\|\nabla f_i(x)\|^2\leq G^2+B^2\|\nabla F(x)\|^2$ follows from bounded gradient assumption. Combining this with (\ref{eq: halfderiv}), we get
 \begin{align*}
     \mathbb{E}&[F(x_t)]\leq F(x_{t-1})\\&+(\eta-\rho)( 6K^2\eta_l^2G^2+6K^2\eta_l^2B^2\|\nabla F(x)\|^2+3K\eta_l^2\sigma^2+\frac{\sigma_g^2}{m}(\frac{n-q-m}{R(n-q)-1}))+\eta V_1+V_3+3\eta\epsilon
 \end{align*}
 Taking $\rho\geq\frac{\alpha\sqrt{\eta}}{6K^2\eta_l^2B^2}+\eta$, we have
 \begin{align*}
     \|\nabla F(x_{t-1})\|^2\leq\frac{\mathbb{E}(F(x_{t-1})-F(x_t)))}{\alpha\sqrt{\eta}}+\frac{\sqrt{\eta}}{\alpha}\mathcal{O}\left(\frac{\sigma_g^2}{m}\left(\frac{n-q-m}{R(n-q)-1}\right)+G^2+\sigma^2+V_1+V_3+\epsilon\right)
 \end{align*}
 Telescoping and using expectation after $T$ global epochs, we get
 \begin{align*}
    \frac{\mathbb{E}[\sum_{t\in[T]}\|\nabla F(x_{t-1})\|^2]}{T}\leq\frac{\mathbb{E}[F]}{\alpha\sqrt{\eta}T}+\frac{\sqrt{\eta}}{\alpha}\mathcal{O}\left(\frac{\sigma_g^2}{m}\left(\frac{n-q-m}{R(n-q)-1}\right)+G^2+\sigma^2+V_1+V_3+\epsilon\right)
\end{align*}
\end{proof}
\section{Security}
To summarize, the security protocol operates in multiple rounds as is the case with any secure aggregation oracle in distributed learning. Firstly, keys are shared among every pair of clients in a cluster, this is followed by collection of masked inputs among each cluster by the server, which are then averaged within the cluster after a consistency check to make sure enough participants have participated in the round. Since model parameters are 32-bit floating points we convert them to integers and perform the masking modulo $2^{32}$.

In particular, each client masks its private update with random vectors such that the server, even if curious, does not learn anything more than the sum of updates from a client cluster. For a given cluster $\mathcal{M}_k,k\in[c]$ assume that $\Delta_i,\sum_{i\in\mathcal{M}_k}\Delta_j\in\mathbb{Z}_P$, for some P.  Consider an order on all the clients within a cluster and each pair of users $i,j (i<j)$ agree on a random vector $r_{i,j}$. If $i$ adds this to its updates ($\Delta_i$) and $j$ subtracts it from its update ($\Delta_j$), adding them would cancel and server would learn just the average but not individual updates. Hence, each client $i\in\mathcal{M}_k$ would compute $u_i=\Delta_i+\sum_{j\in\mathcal{M}_k,i<j}r_{i,j}-\sum_{j\in\mathcal{M}_k,i>j}r_{i,j}$ (mod P). If no clients drop in the computation round, it can be seen that $\sum_{i\in\mathcal{M}_k} u_i=\sum_{i\in\mathcal{M}_k}\Delta_i$
 (mod P). Further, this can be made communication efficient by coordinating a common agreement on seeds for pseudorandom generator.
 
 %In the worst case, a benign client's signal is lost if it is clustered with a malicious client across all reclustering rounds. In particular, the probability that a benign client is effected in a global round is $\left(1-\frac{\binom{n-q}{\frac{n}{m}-1} }{\binom{n}{\frac{n}{m}-1}}\right)^R$. Hence

\subsection{Communication efficiency}
 Notice that sharing a whole mask has a communication cost that increases linearly with the model size and hence prevents dealing with large models. This can be circumvented by clients agreeing on common seeds for a pseudorandom generator (PRG). PRG takes in a random seed as input and generates uniformly random numbers in $[0,P)^d$ where $d$ is the model dimension. Engaging in a key agreement after broadcasting Diffie–Hellman public keys, as stated in \cite{bonawitz2017practical} is a way to compute these shared seeds. Hence, each client belonging to a cluster $i\in\mathcal{M}_k$ would compute $u_i=\Delta_i+\sum_{j\in\mathcal{M}_k,i<j}\text{PRG}(s_{i,j})-\sum_{j\in\mathcal{M}_k,i>j}\text{PRG}(s_{i,j})_{i,j}$ (mod P), where $s_{i,j}$ is the shared seed between clients $i,j$.
 
\subsection{Handling dropped users}
Since masks are shared between clients in a cluster, dropping of a client in a round causes incorrect computation of average as the masks do not exactly cancel each other. This problem is resolved by utilizing Samir's t out of n Secret Sharing \cite{shamir1979share} to share each clients' Diffie–Hellman secret with others and hence server can retrieve masks for the dropped client. Optionally, double masking as noted in \citet{bonawitz2017practical} can be used to enhance security.
\subsection{Privacy at Server}
As noted in Section \ref{Sec: Theory}, server learns nothing more than the average of the clustered client updates. However, multiple reclustering rounds poses an additional privacy threat since multiple averages among same clients appear in clear to the server. We note that the server requires at least $R\geq \frac{n}{c}$ to identify all the updates, this can be used to tune $R,c$. Further, $R\geq \frac{n}{c}$ does not guarantee that server learns all updates since clusterings can overlap resulting in linearly dependent equations with infinite solutions.  

\subsection{Privacy from curious clients}
As mentioned in Section \ref{Sec: Theory}, at least $m-1$ malicious clients are required in a cluster to infer the update of the remaining client. In each reclustering round this probability is upper bounded by $\mathcal{O}(\frac{m(n-m)!}{(q-1-m)!})$ and hence the rest being constant, as cluster size increases, it gets harder to break a client's privacy. Further, we note that just as $n-1$ colluding curious clients can break privacy in traditional Federated Learning, $m-1$ colluding curious clients can in our approach.
\subsection{Communication costs}
\paragraph{Server:} The server communication cost is $\mathcal{O}(Rnm+Rdn)$, where $R,m,n,d$ indicate number of reclustering rounds, number of clients per cluster, total number of clients and model size respectively. Here $\mathcal{O}(Rmn)$ is associated with mediation of pairwise communication between clients in each cluster and $\mathcal{O}(Rdn)$ is for receiving masked data vectors from each user. Although reclustering increases communication costs, clustering helps reduce pairwise communications.
\paragraph{Client:} Client communication cost is $\mathcal{O}(Rm+Rd)$. Here $\mathcal{O}(Rm)$ is associated with pairwise key exchange within a cluster over all reclustering rounds and $\mathcal{O}(Rd)$ is for communicating its model to server in every reclustering round. 

Note that these are passive adversaries hence while can be curious, they honestly follow the protocol for security. Compared to the two server approach suggested in \citet{he2020secure}, our approach can handle attacks on the server, because if an adversary attacks the server(s) or can see communication channels, model updates still remain private.
\section{Additional Experiments}
We use a global learning rate $\eta=1$, local learning rate $\eta_l=0.01$, local momentum 0.9, and mini-batch size 64. We run each experiment for 200 global rounds with 2 local rounds each and report top-1 accuracy on the testing set. For all the experiments, unless specified, we use $R=10$ reclustering rounds.

For Zeno++, we randomly sample 5\% of the training data across clients with the same number of samples of each label to use as the server-side validation set as in~\citet{xie2020zeno++}. We consider batch size of 128, $\rho=0.0001,\epsilon=0.2$ as Zeno++ parameters.
 
In experiments on Shakespear, we use $\eta=1$, $\eta_l=1$, local momentum 0.9, and mini-batch size 256. We run each experiment for 100 global rounds with 2 local rounds each. For all the experiments, unless specified, we use $R=10$ reclustering rounds. Codes for the same would be made available soon.
%All the codes used and experiment results are available at \url{https://github.com/RajKiriti/two_step_federated_aggregation/tree/paper-release}
\subsection{Impact of cluster size}
\subsubsection{Fall of Empires Attack}
 We now test the sensitivity of cluster size on Byzantine-tolerance to a stronger attack with colluding adversaries, we utilize a modified version of Fall of Empires (FoE) \cite{pmlr-v115-xie20a}. In particular, each malicious client sends a negatively scaled averaged model update across all malicious clients. We test scaling these updates by $\beta=-1,-10$. Since Zeno++ can tolerate a greater fraction of clients being Byzantine, we set $q=6$ while for trimmed mean and Krum, we set $q=3$. (Parameters for defenses remain the same as in Section \ref{Sec: Experiments} unless specified.)
 
 \begin{figure}[h]
\centering
\subcaptionbox{Zeno++, $\beta=-1,q=6$\label{zeno-g}}
    {\includegraphics[width=0.32\textwidth]{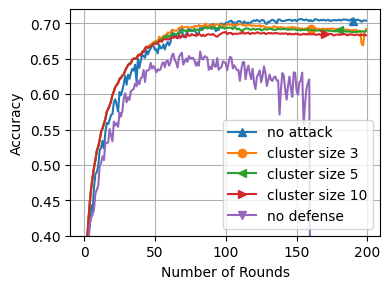}}
\subcaptionbox{Krum, $\beta=-1,q=3$\label{krum-g}}
    {\includegraphics[width=0.32\textwidth]{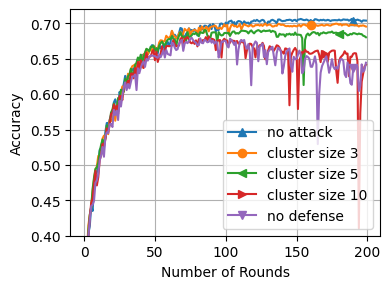}}
\subcaptionbox{Trmean, $\beta=-1,q=3$\label{trimmed-mean-g}}
    {\includegraphics[width=0.32\textwidth]{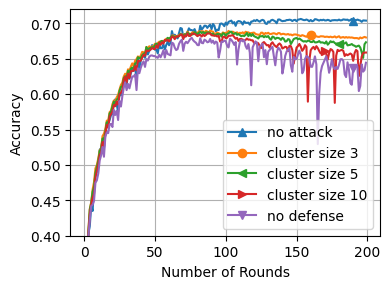}}

\subcaptionbox{Zeno++, $\beta=-10,q=6$\label{zeno-10g}}
    {\includegraphics[width=0.32\textwidth]{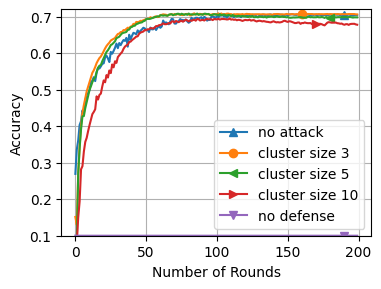}}
\subcaptionbox{Krum, $\beta=-10,q=3$\label{krum-10g}}
    {\includegraphics[width=0.32\textwidth]{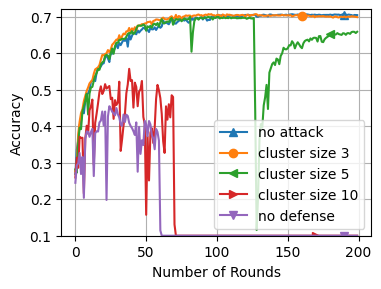}}
\subcaptionbox{Trmean, $\beta=-10,q=3$\label{trimmed-mean-10g}}
    {\includegraphics[width=0.32\textwidth]{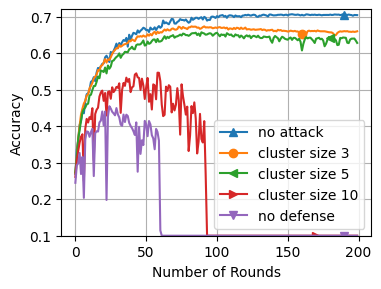}}
\caption{Results of SHARE on CIFAR-10, with varying cluster size across various defenses under FoE attack and $q$ malicious clients out of 60 total clients. Malicious clients send their average update scaled by $\beta$ as indicated in the subfigures. For trimmed mean (Trmean) we remove 2/3 of input updates and use batch size of 128, $\rho=0.00001, \epsilon=0.2$ as Zeno++ parameters.\label{fig: FoECifar}}
\end{figure}
 The results are plotted in Figure \ref{fig: FoECifar}. A similar effect of cluster sizes as discussed in Section \ref{Sec: Experiments} can be seen here. Further, Zeno++ ,as expected, is stable even at higher levels of corruption.
\subsection{Effect of Reclustering with Non-IID data} 
Empirically, we test the effect of reclustering on a heterogeneous data distribution. In particular, we divide CIFAR-10 among clients such that each one gets data from a few classes. In all experiments we use cluster size of 3.
\subsubsection{No Attack}
To create heterogeneous data effect, we split CIFAR-10 data across 60 clients such that each client gets only data consisting of 2 labels. As can be seen in Figure \ref{fig: defresamp}, robust defenses such as coordinate wise median\cite{pillutla2019robust} and Krum\cite{blanchard2017machine} fail in this setting, while SHARE with these defenses performs better as the number of reclustering rounds increases. Further, variance reduction is observed as we increase reclustering rounds ($R$). 
\begin{figure}[h]
\centering
\subcaptionbox{Median \label{2shardmed}}
    {\includegraphics[width=0.43\textwidth]{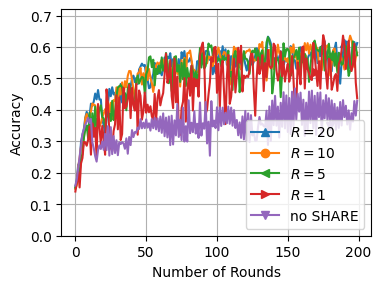}}
\subcaptionbox{Krum \label{krum-g}}
    {\includegraphics[width=0.43\textwidth]{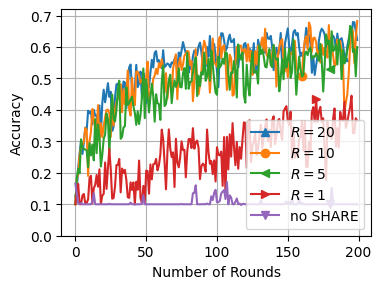}}
\caption{Results of SHARE on non-IID CIFAR-10, where each client has data from 2 labels, across various defenses without attack with 50 clients. We vary the number of reclustering rounds ($R$). No SHARE indicates the baseline defense without the SHARE framework. \label{fig: defresamp}}
\end{figure}
\subsubsection{Fall of Empires Attack}
We consider a lower level of heterogeneity but with client corruption. In particular, each client gets data consisting of 5 labels and $q=3$ malicious clients which collude to send their average model update scaled by $\beta=-10$. The results are shown in Figure \ref{fig: Foeresamp} for various number of reclustering rounds ($R$). 
\begin{figure}[h]
\centering
\subcaptionbox{Median \label{2shardmed}}
    {\includegraphics[width=0.43\textwidth]{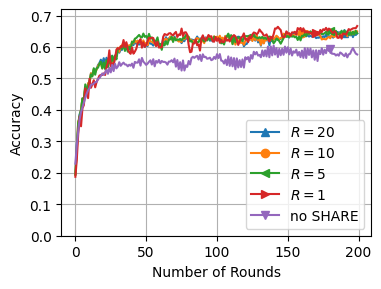}}
\subcaptionbox{Krum \label{krum-g}}
    {\includegraphics[width=0.43\textwidth]{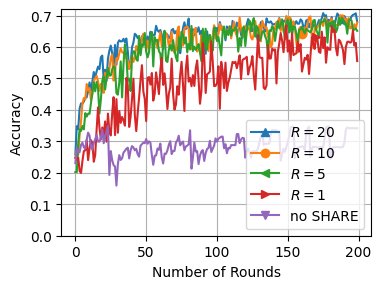}}
\caption{Results of SHARE on non-IID CIFAR-10, where each client has data from 5 labels, across various defenses with FoE attack ($\beta=-10$) and $q=3$ malicious clients out of 50 total. We vary the number of reclustering rounds ($R$). No SHARE indicates the baseline defense without the SHARE framework.\label{fig: Foeresamp}}
\end{figure}

\subsubsection{Label Flip Attack}
We test the effect of SHARE on the label-flip attack with a heterogeneous data split of CIFAR-10. Each client gets data from 5 labels. Malicious clients train on flipped labels, i.e. any label $\in\{0,\hdots,9\}$ is changed to $9-$label. We test trimmed mean (filtering $2/3$ of input updates) and Zeno++ with batch size of 128, $\rho=0.00001, \epsilon=1$ with SHARE and use $\epsilon=5$ without SHARE framework. These parameters are tuned to achieve good performances within their respective frameworks. We consider $R=10$ reclustering rounds and $q=12$ malicious clients. Results as shown in Figure \ref{fig: lfresamp} demonstrate the efficacy of SHARE. 
\begin{figure}[h]
\centering
\subcaptionbox{Trimmed mean \label{2shardmed}}
    {\includegraphics[width=0.43\textwidth]{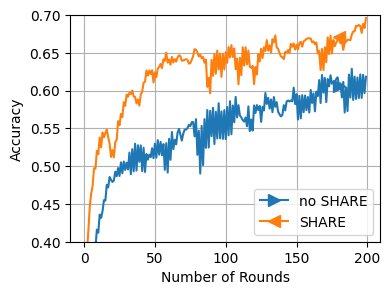}}
\subcaptionbox{Zeno++ \label{krum-g}}
    {\includegraphics[width=0.43\textwidth]{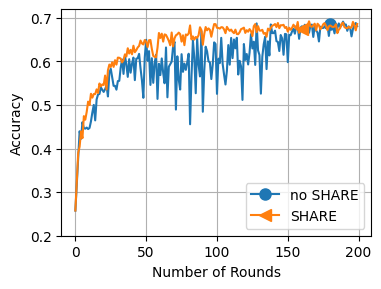}}
\caption{Results of SHARE on non-IID CIFAR-10, where each client has data from 5 labels, across various defenses with label-flip attack, $q=12$ malicious clients out of 60 total, and $R=10$ reclustering rounds. No SHARE indicates the baseline defense without the SHARE framework. \label{fig: lfresamp}}
\end{figure}
\subsubsection{Fixed clustering vs reclustering}
Figure \ref{fig: FoEclust} compares results between fixed clustering and SHARE. In the former, we fix the client clusters before the learning process starts while the latter allows for random reclustering in every round. CIFAR-10 data is split heterogeneously such that each client receives 5 class labels alone. We use Fall of Empires attack with $\beta=-10$ scaling of the average gradients from the malicious clients. Since Zeno++ can tolerate higher levels of corruption, we consider $q=6$ malicious clients, while for trimmed mean, we consider $q=3$. We use $R=10$ reclustering rounds and 60 clients in total with a cluster size of 3. We test trimmed mean (filtering $2/3$ of input updates) and Zeno++ with batch size of 128, $\rho=0.00001, \epsilon=1$ with SHARE and use $\epsilon=5$ for fixed clustering case. These parameters are tuned to achieve reasonable performances within their respective frameworks.
\begin{figure}[h]
\centering
\subcaptionbox{Trimmed mean, $q=3$ \label{2shardmed}}
    {\includegraphics[width=0.43\textwidth]{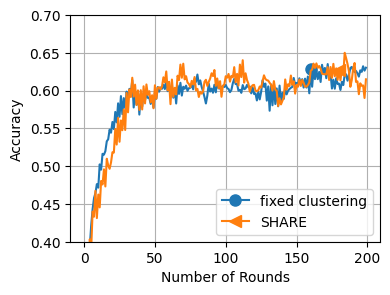}}
\subcaptionbox{Zeno++, $q=6$ \label{krum-g}}
    {\includegraphics[width=0.43\textwidth]{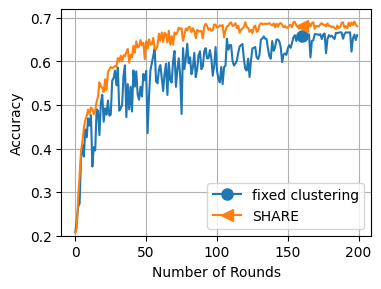}}
\caption{Results of SHARE on non-IID CIFAR-10, where each client has data from 5 labels, across various defenses with FoE attack ($\beta=-10$) and $q$ malicious clients out of 60 total. Fixed clustering indicates the baseline defense with fixed clusters of size 5. \label{fig: FoEclust}}
\end{figure}

\section{Random Reclustering}
In the worst case, a benign client's signal is lost if it is clustered with a malicious client across all reclustering rounds. In particular, the probability that a benign client is effected in a global round is $\left(1-\frac{\binom{n-q}{\frac{n}{m}-1} }{\binom{n}{\frac{n}{m}-1}}\right)^R$. Hence this probability decreases as R (number of reclustering rounds increase).

%\paragraph{Curious colluding clients:} Setting a higher threshold $t$ allows for higher levels of corruption among a cluster, but this can be at most $m-1$, where $m$ denotes cluster size. Hence if a cluster has $m-1$ adversaries which can intercept communication of clients within the cluster and collude, it poses a privacy threat to the remaining benign user. Hence a maximum of $\lfloor\frac{q}{m-1}\rfloor$ benign client updates can be intercepted. This illustrates that for fixed number of malicious clients, increasing cluster sizes makes it harder for this security breach. 

%We utilize the following theorem to prove it's convergence. (Zeno++)

%This image classification dataset consists of 50$k$ train images and 10$k$ test images equally spread among 10 classes.
%Optionally include extra information (complete proofs, additional experiments and plots) in the appendix.
%This section will often be part of the supplemental material.

\end{document}